\documentclass[conference]{IEEEtran}%

\normalsize

 \usepackage{cite}
\usepackage[]{graphicx}
\graphicspath{{figures/}}
\DeclareGraphicsExtensions{.eps,.pdf}
\usepackage[mathscr]{euscript}
\usepackage{balance}

\usepackage{hyperref}
\usepackage[usenames,dvipsnames]{color}

\usepackage{graphicx}
\usepackage{subfig} 
\usepackage{latexsym}
\usepackage{amssymb}
\usepackage{verbatim}
\usepackage{amsmath}
\usepackage{amsthm}
\usepackage{amsfonts}
\usepackage{stackengine}
\usepackage{bm} 
\usepackage{float}
\usepackage{stfloats}
\usepackage{mathtools}
\usepackage{dsfont}
 
\usepackage{multicol}

\newtheorem{Theorem}{Theorem}

\newtheorem{Remark}{Remark}

\newtheorem{Example}{Example}

\IEEEoverridecommandlockouts

\begin{document}
%
\title{Distributed and Private Coded Matrix Computation with Flexible Communication Load}
%
%
%

\author{\IEEEauthorblockN{\textbf{Malihe Aliasgari}\IEEEauthorrefmark{1},
\textbf{Osvaldo Simeone}\IEEEauthorrefmark{2}
and \textbf{J\"org Kliewer}\IEEEauthorrefmark{1}} 
\IEEEauthorrefmark{1}Department of Electrical and Computer
Engineering, New Jersey Institute of Technology,  Newark, NJ, U.S.A.\\
\IEEEauthorrefmark{2}King's College London, Department of Informatics,  London, U.K.
\thanks{This work was supported in part by 
the European Research Council (ERC) under
the European Union Horizon 2020 research and innovative programme (grant
agreement No 725731)
and by U.S. NSF grants CNS-1526547, CCF-1525629. O. Simeone is currently on leave from NJIT.

}}



\maketitle

\begin{abstract}
Tensor operations, such as matrix multiplication, are central to large-scale
machine learning applications. For user-driven tasks these operations can be
carried out on a distributed computing platform with a master server at the
user side and multiple workers in the cloud operating in parallel. For
distributed platforms, it has been recently shown that coding over the input
data matrices can reduce the computational delay, yielding a trade-off
between recovery threshold  and communication load. In this paper we impose
an additional security constraint on the data matrices and assume that
workers can collude to eavesdrop on the content of these data
matrices. Specifically, we introduce  a novel class of secure codes,
referred to as secure generalized PolyDot codes, that
generalizes previously published non-secure versions of these codes for
matrix multiplication. These codes extend the state-of-the-art by allowing a
flexible trade-off between recovery threshold and communication load for a
fixed maximum number of colluding workers.
\end{abstract}
 
\begin{IEEEkeywords}
Coded distributed computation, distributed learning, secret sharing,
information theoretic security.
\end{IEEEkeywords}

 \IEEEpeerreviewmaketitle

\section{Introduction}

At the core of many signal processing and machine learning applications are
tensor operations such as matrix multiplications
\cite{janzamin2015beating}. In the presence of practically sized data sets,
such operations are typically carried out using distributed computing
platforms with a master server and multiple workers that can operate in
parallel over distinct parts of the data set. The master server plays the
role of the parameter server, distributing data to the workers and periodically reconciling their internal state \cite{li2014scaling}. Workers are commercial off-the-shelf servers that are characterized by possible temporary failures and delays \cite{dean2013tail}. While current distributed computing platforms conventionally handle straggling servers by means of replication of computing tasks \cite{huang1984algorithm}, recent work has shown that encoding the input data can help reduce the computation latency, which depends on the number of tolerated stragglers by orders of magnitude, e.g., \cite{joshi2017efficient,wang2015using}.
More generally, coding is able to  control the trade-off between
computational delay and communication load between workers and master server
\cite{lee2018speeding,yu2017polynomial,li2018fundamental,aliasgari2018coded,aliasgari2018codedisit}.
Furthermore, stochastic coding can help keeping both input and output data
secure from the workers, assuming that the latter are honest by carrying out the prescribed protocol, but curious \cite{nodehi2018limited,yu2018lagrange,chang2018capacity,kakar2018rate,yang2019secure,rafael2018codes}. This paper contributes to this line of work by investigating the trade-off between computational delay and communication load as a function of the privacy level (see Fig.~\ref{figSecRT_CO} for a preview).

As illustrated in Fig.~\ref{FigComputSyst}, we focus on the basic problem of computing the matrix multiplication $\mathbf{C}=\mathbf{AB}$ in a distributed computing system of $P$ workers that can process each only a fraction $1/m$ and $1/n$ of matrices $\mathbf{A}$ and $\mathbf{B}$, respectively. Three performance criteria are of interest: (\emph{i}) the recovery threshold $P_R$, that is, the number of workers that need to complete their task before the master server can recover the product $\mathbf{C}$; (\emph{ii}) the communication load $C_L$ between workers and master server; and (\emph{iii}) the maximum number $P_C$ of colluding servers that ensures perfect secrecy for both data matrices $\mathbf{A}$ and $\mathbf{B}$. In order to put our contribution in perspective, we briefly review next prior related work.

Consider first solutions that provide no security guarantees, i.e.,
$P_C=0$. As a direct extension of \cite{lee2018speeding}, a first approach
is to use product codes that apply separate MDS codes to encode the two
matrices \cite{lee2017high}. The recovery threshold of this scheme is
improved by \cite{yu2017polynomial} which introduces so called polynomial
codes.
The construction in \cite{yu2017polynomial} is proved to be optimal
under the assumption that minimal communication is allowed between workers
and master server. In \cite{fahim2017optimal} so called MatDot codes are
introduced,
resulting in a lower recovery threshold at the expense
of a larger communication load. The construction in \cite{dutta2018optimal} bridges
the gap between polynomial and MatDot codes and presents so called PolyDot
codes,
yielding a trade-off between recovery threshold and
communication load. An extension of this scheme, termed Generalized PolyDot
(GPD) codes improves on the recovery threshold of PolyDot codes
\cite{dutta2018unified}, which is independently obtained by the
construction~in~\cite{yu2018straggler}.

Much less work has been done in the
literature if security constraints are factored in, i.e., if $P_C\ne 0$. In
\cite{yu2018lagrange} Lagrange coding is presented which achieves the
minimum recovery threshold for multilinear functions by generalizing MatDot
codes. In \cite{chang2018capacity,kakar2018rate,rafael2018codes} a reduction
of the communication load is addressed by extending polynomial codes. While
these works focus on either minimizing recovery threshold or communication
load, the \emph{trade-off} between these two fundamental quantities has not
been addressed in the open literature to the best of our knowledge. In this
paper, we intend to fill this void and present a novel class of secure
computation codes, referred to as secure GPD (SGPD) codes, that generalize
GPD codes at all communication load levels, yielding a new achievable trade-off between
recovery threshold and communication load as a function of the desired
privacy level.	
\section{System Model}\label{secModel}
\textbf{Notation:}  Throughout the paper, we denote a matrix with upper boldface letters (e.g., $\mathbf{X}$) and lower boldface letters indicate a vector or a sequence of matrices (e.g., $\mathbf{x}$).
 Furthermore, math calligraphic font refers to a set (e.g., $\mathcal{X})$.
 A set $\mathbb{F}$ represents the Galois field with cardinality $|\mathbb{F}|$. For any real number $a$, $\lceil a\rceil$ represents the largest integer nearest to $a$.

 \begin{figure*}[t!]
	\begin{center}
	\includegraphics[scale=0.25]{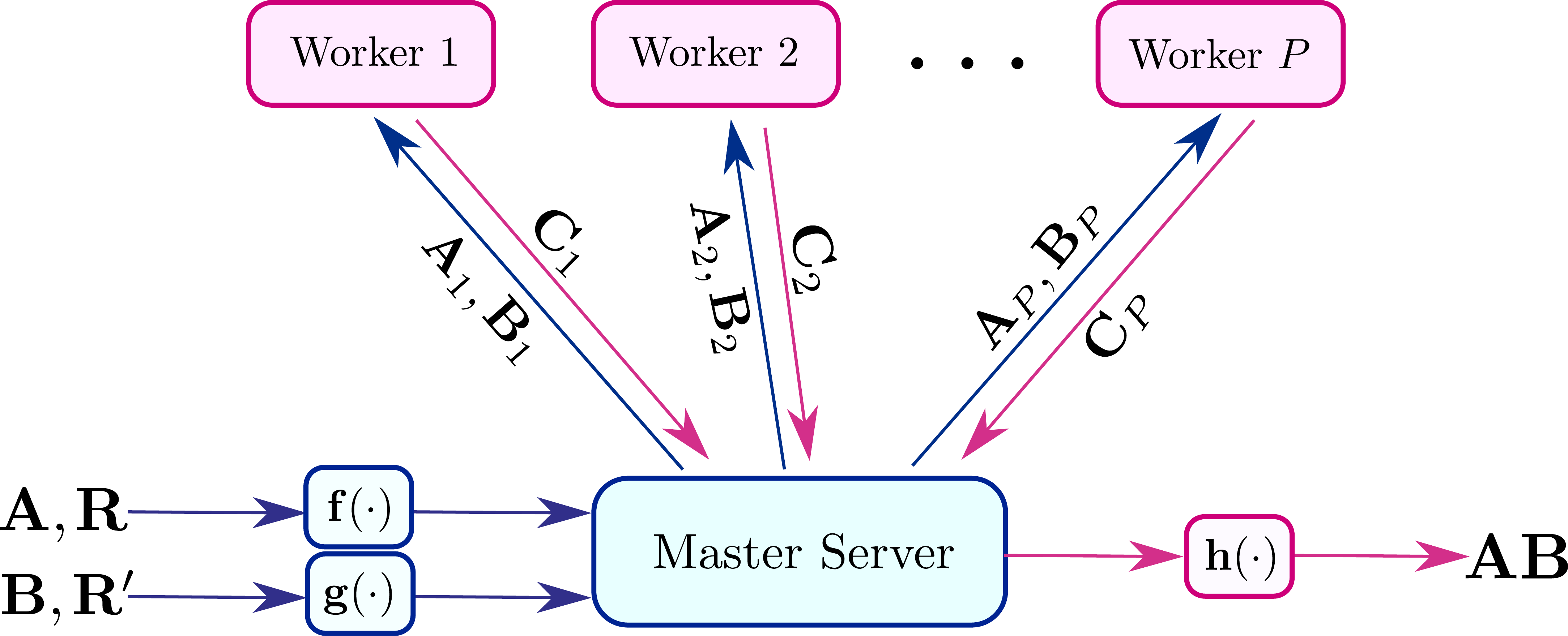}\vspace{-.1cm}~\caption{\footnotesize{ The master server encodes the input matrices $\mathbf A$ and $\mathbf B$ and random matrices $\mathbf{R}$ and $\mathbf{R}'$, respectively, to define the computational tasks of the slave servers or workers. The workers may fail or straggle, and they are honest but curious, with colluding subsets of workers of size at most $P_C$. The master server must be able to decode the product $\mathbf{AB}$ from the output of a subset of $P_R$ servers, which defines the recovery threshold.}}~\label{FigComputSyst}
	\end{center}
	\vspace{-5ex}
\end{figure*}

As illustrated in Fig.~\ref{FigComputSyst}, we consider a distributed computing system with a master server and $P$ slave servers or workers. The master server is interested in computing securely the matrix product $\mathbf{C}=\mathbf{AB}$ of two data matrices $\mathbf{A}$ and $\mathbf{B}$ with dimensions $T\times S$ and $S\times D$, respectively. The matrices have entries from a sufficient large finite field $\mathbb{F}$, with $|\mathbb{F}|>P$. Both matrices $\mathbf{A}$ and $\mathbf{B}$ contain confidential data. The $P$ workers receive information on matrices $\mathbf{A}$ and $\mathbf{B}$ from the master; they process this information; and they respond to the master, which finally recovers the product $\mathbf{AB}$ with minimal computational effort. Each worker can receive and process only $TS/m$ and $SD/n$ symbols, respectively, for some integers $m$ and $n$. 
The workers are honest but curious. Accordingly, we impose the secrecy constraint that, even if up to $P_C < P$ workers collude, the workers cannot obtain any information about both matrices $\mathbf A$ and $\mathbf B$  based on the data received from the master server. 

To keep the data secure and to leverage possible computational redundancy at the workers (namely, if $P/m>1$ and/or $P/n>1$), the master server sends encoded versions of the input matrices to the workers. 
Due to the above mentioned communication and storage constraints, the encoded matrices $ {\mathbf A}_p=\mathbf f_{p}(\mathbf A,\mathbf{R})$, with $\mathbf f_{p}:\mathbb{F}^{TS/m}\times \mathbb{F}^{TS/m}\rightarrow\mathbb{F}^{TS/m}$,  and $ {\mathbf B}_p=\mathbf g_{p}(\mathbf B,\mathbf{R}')$, with $\mathbf g_{p}:\mathbb{F}^{SD/n}\times \mathbb{F}^{SD/n}\rightarrow\mathbb{F}^{SD/n}$, 
to be sent to each $p$th worker, $p=1,\ldots,P$, have $TS/m$ and $SD/n$ entries, respectively, for some encoding functions $\mathbf{f}_p(\cdot)$ and $\mathbf{g}_p(\cdot)$. The random matrices $\mathbf{R}$ and $\mathbf{R}'$ consisting an arbitrary number of uniform i.i.d. randomly distributed entries from a field $\mathbb{F}$.
The security constraint imposes the condition
\begin{equation}\label{eq_security_constraint}
I( {\mathbf A}_{\mathcal{P}}, {\mathbf B}_{\mathcal{P}};\mathbf A,\mathbf B)=0, 
\end{equation}
for all subsets of $\mathcal{P}\subset [1,P]$ of $P_C$ workers, where the random matrices $\mathbf{R}$ and $\mathbf{R}'$ serve as a form of random keys in order to meet the security constraint \eqref{eq_security_constraint} \cite{shamir1979share}.

Each worker $p$ computes the product $\mathbf C_p= {\mathbf A}_p {\mathbf B}_p$ of the encoded sub-matrices $ {\mathbf A}_p$ and $ {\mathbf B}_p$. 
The master server collects a subset of $P_R\leq P$ outputs from the workers as defined by the subset $\{\mathbf{C}_p\}_{p \in \mathcal{P}_R}$ with $|\mathcal{P}_R|=P_R$. It then applies a decoding function $\mathbf{h}\left(\{\mathbf C_p\}_{p\in \mathcal{P}_R}\right)$, $\mathbf{h}:\underbrace{\mathbb{F}^{TD/td}\times\cdots \times \mathbb{F}^{TD/td}}_{P_R \text{~times}}\rightarrow \mathbb{F}^{TD}$. 
Correct decoding translates into the condition 
\begin{equation}\label{eq_decidability_constraint}
H(\mathbf{AB}|\{\mathbf C_p\}_{p\in \mathcal{P}_R})=0. 
\end{equation}
For given  parameters $m$ and $n$, the performance of a coding and decoding scheme is measured by the triple $(P_C,P_R,C_L)$, where $P_C$ is the maximum number of colluding workers; $P_R$ is the recovery threshold, i.e., the minimum number of workers whose outputs are used by the master to recover the product $\mathbf{AB}$; and $C_L$ is the communication load defined as $C_L=\sum_{p\in \mathcal{P}_R}|\mathbf{C}_p|$. Here, $|\mathbf{C}_p|$ is the dimension of the product matrix $\mathbf{C}_p$ computed by worker $p$.
Note that condition \eqref{eq_decidability_constraint} requires  the inequality 
$\min\{P_R/m,P_R/n\}\geq 1$ or
$P_R\geq P_{R,\min} \overset{\Delta}{=} \max \{m,n\}$, which is hence a lower bound for the minimum recovery threshold. Furthermore, the communication load is lower bounded by $C_L\geq C_{L,\min}\overset{\Delta}{=}TD$, which is the size of the product $\mathbf{C}=\mathbf{AB}$.

 \subsection{Generalized PolyDot Code  without Security Constraint}\label{Sec_GPD} 
 
 In this subsection, we review the GPD construction first proposed in
 \cite{fahim2017optimal} and later improved in \cite{yu2018straggler,dutta2018unified}. 
 This coding scheme  achieves the best currently known trade-off between recovery threshold $P_R$ and communication load $C_L$ for $P_C=0$, i.e., under no security constraint. The equivalent entangled polynomial codes of \cite{yu2018straggler} have the same properties in terms of $(P_R,P_C)$.
 The GPD codes for $P_C=0$ also achieve the optimal recovery threshold among all linear coding strategies in the cases of $t = 1$ or $d = 1$, also they minimize the recovery threshold for the minimum communication load $C_{L,\min}$ \cite{yu2017polynomial,yu2018straggler}. 
  
The GPD code splits the data matrices $\mathbf{A}$ and $\mathbf{B}$ both horizontally and vertically as
\begin{equation}\label{Polydot}
{ 
	\mathbf{A}=
	\left[ {\begin{array}{cccc}
		\mathbf{A}_{1,1}&\ldots&\mathbf{A}_{1,s}\\
		\vdots&\ddots&\vdots\\
		\mathbf{A}_{t,1}&\ldots&\mathbf{A}_{t,s}
		\end{array} } \right],}
\quad
\mathbf{B}=
\left[ {\begin{array}{cccc}
	\mathbf{B}_{1,1}&\ldots&\mathbf{B}_{1,d}\\
	\vdots&\ddots&\vdots\\
	\mathbf{B}_{s,1}&\ldots&\mathbf{B}_{s,d}
	\end{array} } \right].
\end{equation}
The parameters $s,t$, and $d$ can be set arbitrarily under the constraints $m=ts$ and $n=sd$. Note that polynomial codes set $s=1$, while MatDot codes have $t=d=1$ \cite{dutta2018optimal}.
All  sub-matrices $\mathbf{A}_{ij}$ and $\mathbf{B}_{kl}$ have dimensions $T/t\times S/s$ and $S/s\times D/d$, respectively.
\begin{figure*}[t!]
	\begin{center}
		\includegraphics[width=9.9cm]{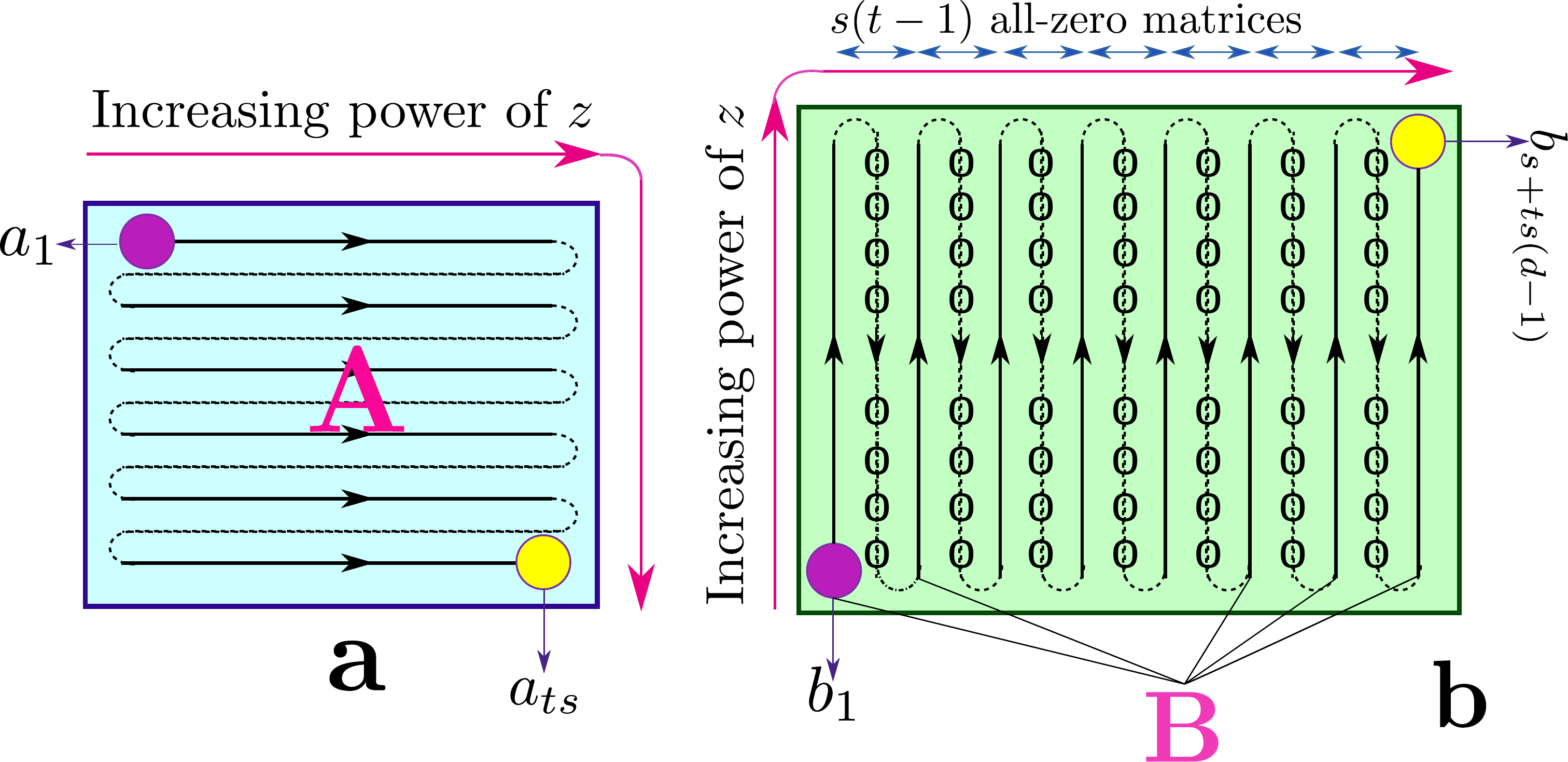}~\caption{\footnotesize 
		Construction of the time sequences $\mathbf{a}$ and $\mathbf{b}$ used to define the generalized PolyDot (GPD) code. The zero dash line in  $\mathbf{b}$ indicates all-zero block sequences. Each solid arrows in $\mathbf{a}$ and $\mathbf{b}$ shows a distinct row of $\mathbf{A}$ and a column of $\mathbf{B}$, respectively.  }~\label{figGPD}
	\end{center}
	\vspace{-5ex}
\end{figure*}
The GPD code computes each block $(i,j)$ of the product $\mathbf{C}=\mathbf{AB}$, namely $\mathbf{C}_{i,j}=\sum_{k=1}^s\mathbf{A}_{i,k}\mathbf{B}_{k,j}$, for $i=1,\ldots,t$ and $j=1,\ldots,d$, in a distributed fashion. 
This is done by means of polynomial encoding and polynomial interpolation. As we review next, the computation of block $\mathbf{C}_{i,j}$ can be interpreted as the evaluation of the middle sample of the convolution $\mathbf{c}_{i,j}=\mathbf{a}_i*\mathbf{b}_j$ between the block sequences $\mathbf{a}_i=[\mathbf{A}_{i,1},\ldots,\mathbf{A}_{i,s}]$ and $\mathbf{b}_j=[\mathbf{B}_{s,j},\ldots,\mathbf{B}_{1,j}]$.
In fact, the $s$th sample of the block sequence $\mathbf{c}_{i,j} $ equals $\mathbf{C}_{i,j}$, i.e., $[\mathbf{c}_{i,j}]_s=\mathbf{C}_{i,j}$.
The computation is carried out distributively in the frequency,  domain by using $z$-transforms with different workers being assigned distinct samples in the frequency domain.

To elaborate, define the block sequence $\mathbf{a}$ obtained by concatenating the block sequences $\mathbf{a}_i$ as $\mathbf{a}=\{\mathbf{a}_1,\mathbf{a}_2,\ldots,\mathbf{a}_t\}$. 
Pictorially, a sequence $\mathbf{a}$ is obtained from the matrix $\mathbf{A}$ by reading the blocks in the left-to-right top-to-bottom order, as seen in Fig.~\ref{figGPD}.
We also introduce the longer time block sequence $\mathbf{b}$ as 
\begin{equation}\label{eqTimeSeqB}
\mathbf{b}=\{\mathbf{b}_1,\mathbf{0},\mathbf{b}_2,\mathbf{0},\ldots,\mathbf{b}_d\},
\end{equation}
with ${\mathbf{0}}$ being a block sequence of $s(t^*-1)$ all-zero block matrices with dimensions $S/s\times D/d$. The sequence $\mathbf{b}$ can be obtained from matrix $\mathbf{B}$ by following the bottom-to-top left-to-right order shown in Fig.~\ref{figGPD} and by adding the all-zero block sequences between any two columns of the matrix  $\mathbf{B}$.  

In the frequency domain, the $z$-transforms of sequences $\mathbf{a}$ and $\mathbf{b}$ are obtained as
\begin{align}
\mathbf{F}_{\mathbf{a}}(z)=&\sum_{r=0}^{ts-1} [\mathbf{a}]_{r+1}z^{r}=\sum_{i=1}^t\sum_{j=1}^s\mathbf{A}_{i,j}z^{s(i-1)+j-1},\label{eqDuttapAz}\\
\mathbf{F}_{\mathbf{b}}(z)=&\sum_{r=0}^{s-1+ts(d-1)}[\mathbf{b}]_{r+1}z^{r}=\sum_{k=1}^s\sum_{l=1}^d\mathbf{B}_{k,l}z^{s-k+ts(l-1)},\label{eqDuttapBz}
\end{align}
respectively.
The master server evaluates the polynomials $\mathbf{F}_{\mathbf{a}}(z)$ and $\mathbf{F}_{\mathbf{b}}(z)$ in $P$ non-zero distinct points $z_1,\ldots z_P\in \mathbb{F}$ and sends the corresponding linearly encoded matrices $\mathbf{A}_p=\mathbf{F}_{\mathbf{a}}(z_p)$ and $\mathbf{B}_p=\mathbf{F}_{\mathbf{b}}(z_p)$ to server $p$. The encoding functions are hence given by the polynomial evaluations \eqref{eqDuttapAz} and \eqref{eqDuttapBz}, for $z_1,\ldots,z_p$. Server $p$ computes the multiplication $\mathbf{F}_{\mathbf{a}}(z_p)\mathbf{F}_{\mathbf{b}}(z_p)$ and sends it to the master server.
The master server computes the inverse $z$-transform for the received products $\{\mathbf{A}_p\mathbf{B}_p\}_{p\in \mathcal{P}_R} =\{\mathbf{F}_{\mathbf{a}}(z_p)\mathbf{F}_{\mathbf{b}}(z_p)\}_{p\in \mathcal{P}_R}$, obtaining the convolution $\mathbf{a}*\mathbf{b}$.
From the convolution $\mathbf{a}*\mathbf{b}$, we can see that the master server is able to compute all the desired blocks $\mathbf{C}_{i,j}$ by reading the middle samples of the convolutions $\mathbf{c}_{i,j}=\mathbf{a}_i*\mathbf{b}_j$ from samples of the sequence $\mathbf{c}=\mathbf{a}*\mathbf{b}$ in the order $[\mathbf{c}]_{s-1}=\mathbf{C}_{1,1},[\mathbf{c}]_{2s-1}=\mathbf{C}_{2,1},\ldots,[\mathbf{c}]_{ts-1}=\mathbf{C}_{t,1},[\mathbf{c}]_{s-1+t^*s}=\mathbf{C}_{1,2},\ldots,[\mathbf{c}]_{ts-1+t^*s}=\mathbf{C}_{t,2},\ldots$ (see also the Appendix for details). Note that, in particular, the zero block subsequences added to sequence $\mathbf{b}$ ensure that no interference from the other convolutions, $\mathbf{c}_{i',j'}$ affects the middle ($s$th) sample of a convolution $\mathbf{c}_{i,j}$ with $i'\neq i$ and $j'\neq j$.
To carry out the inverse transform, the master server needs to collect as many values $\mathbf{F}_{\mathbf{a}}(z_p)\mathbf{F}_{\mathbf{b}}(z_p)$ as there are samples of the sequence $\mathbf{a}*\mathbf{b}$, yielding the recovery threshold
\begin{equation}\label{eq_RT}
P_R=tsd+s-1.
\end{equation}
Equivalently, in terms of the underlying polynomial interpretation, the master server needs to collect a number of evaluations of the polynomial $\mathbf{F}_{\mathbf{a}}(z)\mathbf{F}_{\mathbf{b}}(z)$ equal to the degree of $\mathbf{F}_a(z)\mathbf{F}_b(z)$ plus one.
This computation is of complexity order $\mathcal{O}(TDP_R\log^2(P_R))$  \cite{dutta2018optimal}.
 Furthermore, the communication load is given as
\begin{equation}\label{eq_CL}
C_L=P_R\frac{TD}{td},
\end{equation}
where $TD/(td)$ is the size of each matrix $\mathbf{F}_{\mathbf{a}}(z)\mathbf{F}_{\mathbf{b}}(z)$.
%

\section{Secure PolyDot Code}
In this section, we propose a novel extension of the GPD code that is able to ensure the secrecy constraint for any $P_C<P$. We also derive the corresponding achievable set of triples $(P_C,P_R,C_L)$. As we will see, the projection of this set onto the plane defined by the condition $P_C=0$ includes the set of pairs $(P_R,C_L)$ in \eqref{eq_RT} and \eqref{eq_CL} obtained by the GPD code \cite{dutta2018unified}.
The proposed secure GPD code (SGPD) augments matrices $\mathbf{A}$ and $\mathbf{B}$ by adding $P_C$ random block matrices to the input matrices $\mathbf{A}$ and $\mathbf{B}$, in a manner similar to prior works  \cite{nodehi2018limited,yu2018lagrange,chang2018capacity,kakar2018rate,rafael2018codes}, 
yielding augmented matrices $\mathbf{A}^*$ and $\mathbf{B}^*$. As we will see, a direct application of the GPD codes to these matrices is suboptimal. 

In contrast, we propose a novel way to construct sequences $\mathbf{a}^*$ and $\mathbf{b}^*$ from matrices $\mathbf{A}^*$ and $\mathbf{B}^*$ that enables the definition of a more efficient code by means of the $z$-transform approach discussed in the previous section. To this end, we follow the design criterion of decreasing the recovery threshold $P_R$ for a given communication load $C_L$. Based on the discussion in the previous section, this goal can be realized by decreasing the length of sequence $\mathbf{c}^*=\mathbf{a}^**\mathbf{b}^*$, which can in turn be ensured by reducing the length of the sequence $\mathbf{b}^*$ for a given length of sequence $\mathbf{a}^*$. We accomplish this objective by {$\emph (i)$} adaptively appending rows \textit{or} columns with random elements to matrix $\mathbf{A}$, and, correspondingly columns \textit{or} rows to $\mathbf{B}$, which can reduce the recovery threshold; and {$\emph (ii)$} modifying the zero padding procedure (see Fig.~\ref{figGPD}) for the construction of sequence $\mathbf{b}^*$. In order to account for point {\emph{(i)}}, we consider separately the two cases $s< t$ and $s\geq t$.

Note that the code $\mathcal{C}_c$
associated with $\mathbf{F}_{\mathbf{a}}(z)\mathbf{F}_{\mathbf{b}}(z)$ is
obtained as the star product of the codes described by
$\mathbf{F}_{\mathbf{a}}(z)$ and $\mathbf{F}_{\mathbf{b}}(z)$
\cite{Hollanti_etal_2017}, and therefore our proposed SGPD scheme can be
interpreted as a secret sharing scheme \cite{shamir1979share}  employing the star product code $\mathcal{C}_c$.
%
%

   \begin{figure*}[t!]
  	\begin{center}
  		\includegraphics[width=11.2cm]{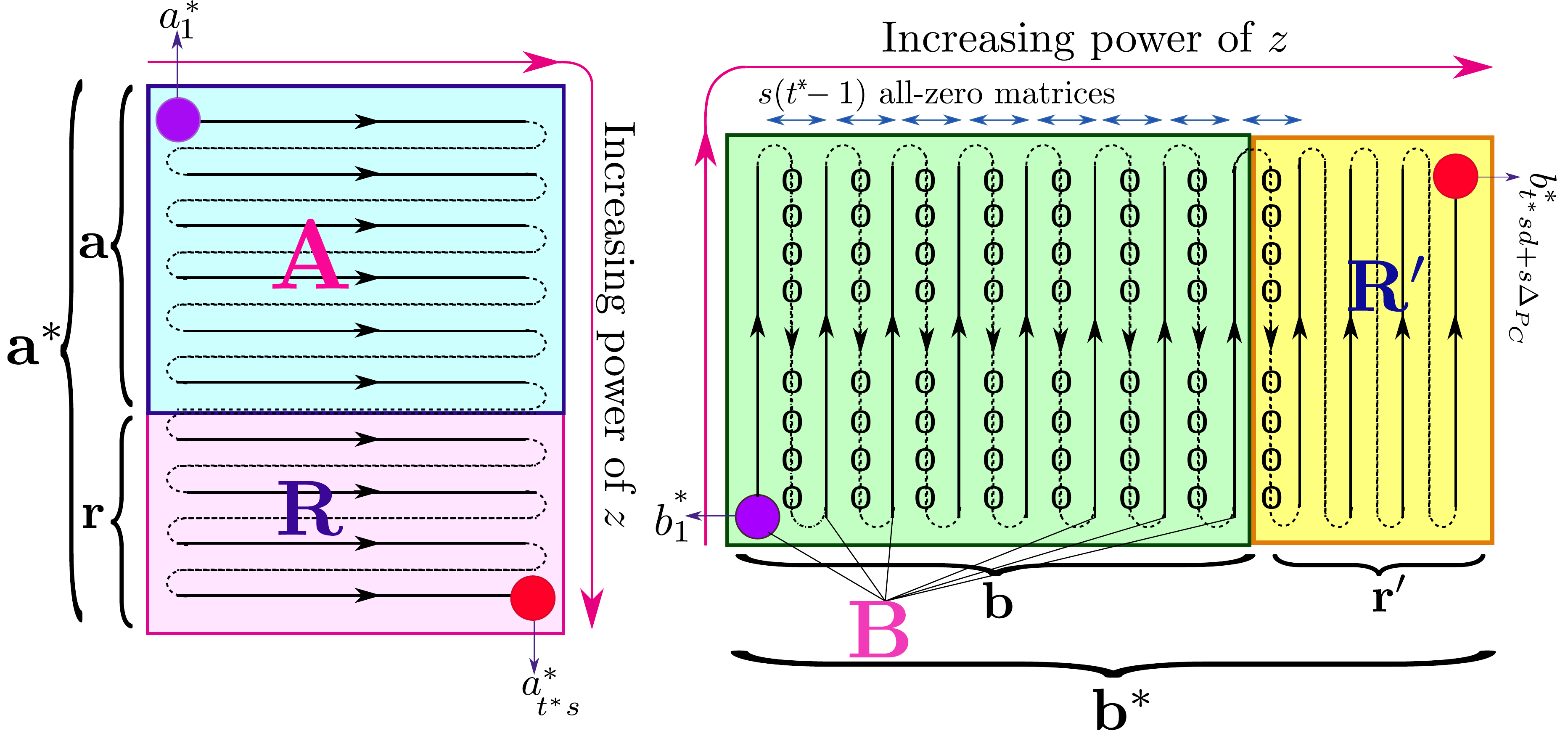}~\caption{\footnotesize
                  Construction of the time block sequences
                  $\mathbf{a}^*=[\mathbf{a},\mathbf{r}]$ and
                  $\mathbf{b}^*=[\mathbf{b},\mathbf{r}']$ in
                  \eqref{eqTimeSeqAR} and \eqref{eqTimeSeqBR} used to define
                  the SGPD code for  the case $s< t$. The zero dashed line in $\mathbf{b}$  and $\mathbf{r}'$ indicate all-zero block sequences.  }~\label{SecGPD} 
  	\end{center}
  	\vspace{-7ex}
  \end{figure*}
\subsection{Secure Generalized PolyDot Code: The $s< t$ Case}

As illustrated in Fig.~\ref{SecGPD}, when $s< t$, we augment input matrices $\mathbf{A}$ and $\mathbf{B}$ by adding 
\begin{equation}\label{eq_delta1}
\Delta_{P_C}\overset{\Delta}{=} \left\lceil 
\frac{P_C}{s} \right\rceil,
\end{equation}
random row and column blocks to matrices $\mathbf{A}$ and $\mathbf{B}$, respectively. Accordingly, the $t^*\times s$ augmented block matrix $\mathbf{A}^*$ with $t^*=t+\Delta_{P_C}$, is obtained as
  \begin{equation}\label{eq_matrixA}
   { 
   	\mathbf{A}^*=\left[ {\begin{array}{c}
   	   		\mathbf{A}\\\mathbf{R} 
   	   		\end{array} } \right]=
   	\left[ {\begin{array}{cccc}
   		\mathbf{A}_{1,1}&\ldots&\mathbf{A}_{1,s}\\
   		\vdots&\ddots&\vdots\\
   		\mathbf{A}_{t,1}&\ldots&\mathbf{A}_{t,s}\\
   		\mathbf{R}_{1,1}&\ldots&\mathbf{R}_{1,s}\\
   		\vdots&\ddots&\vdots\\
   		\mathbf{R}_{\Delta_{{P_C},1}}&\ldots&\mathbf{R}_{\Delta_{{P_C},s}}
   		\end{array} } \right],}
  \end{equation}
  while the $s\times d^*$ augmented matrix $\mathbf{B}^*=[\mathbf{B}~\mathbf{R}']$ with $d^*=d+\Delta_{P_C}$ is obtained as   \begin{equation}\label{eq_matrixB}
  \mathbf{B}^*=
  \left[ {\begin{array}{cccccc}
  	\mathbf{B}_{1,1}&\ldots&\mathbf{B}_{1,d}&\mathbf{R}'_{s,1}&\ldots&\mathbf{R}'_{s,\Delta_{P_C}}\\
  	\vdots&\ddots&\vdots&\vdots&\ddots&\vdots\\
  	\mathbf{B}_{s,1}&\ldots&\mathbf{B}_{s,d}&\mathbf{R}'_{1,1}&\ldots&\mathbf{R}'_{1,\Delta_{P_C}}
  	\end{array} } \right].
  \end{equation} 
  In \eqref{eq_matrixA} and \eqref{eq_matrixB}, if $s$ divides $P_C$, all block matrices $\mathbf{R}_{ij}\in \mathbb{F}^{\frac{T}{t}\times\frac{S}{s}}$ and $\mathbf{R}'_{ij}\in \mathbb{F}^{\frac{S}{s}\times\frac{D}{d}}$ are generated with i.i.d. uniform random elements in $\mathbb{F}$. Otherwise, if $\Delta_{P_C}-P_C/s >0$, the last $s\Delta_{P_C}-P_C$ matrices  in \eqref{eq_matrixA}, with right-to-left ordering in the last row of $\mathbf{R}_{ij}$, and in \eqref{eq_matrixB} with top-to-bottom ordering in the last column of  $\mathbf{R}'_{ij}$, are all-zero block matrices.

As illustrated in Fig.~\ref{SecGPD}, in the SGPD scheme, the block sequence $\mathbf{a}^*$ is defined in the same way as in the conventional GPD, yielding  
  \begin{equation}\label{eqTimeSeqAR}
  \mathbf{a}^*=\{\mathbf{a}_1,\ldots,\mathbf{a}_t,\mathbf{r}_1,\ldots,\mathbf{r}_{\Delta_{P_C}}\},
  \end{equation}
where $\mathbf{r}_i$ is the $i$th row of the block matrix $\mathbf{R}$, $i=1,\ldots,\Delta_{P_C}$. We also define the time block sequence $\mathbf{b}^*=\{\mathbf{b},\mathbf{r}'\}$ as
 \begin{equation}\label{eqTimeSeqBR}
  \mathbf{b}^*=\{\mathbf{b}_1,\mathbf{0},\mathbf{b}_2,\mathbf{0},\ldots,\mathbf{b}_d,\mathbf{0},\mathbf{r}'_1,\mathbf{r}'_2,\ldots,\mathbf{r}'_{\Delta_{P_C}}\},
  \end{equation}
   where $\mathbf{0}$ is block sequences of $s(t^*-1)$ all-zero block matrices, respectively, with dimensions $S/s\times D/d$, while $\mathbf{r}'_j$ is the $j$th column of the random matrix $\mathbf{R}'$. The key novel idea of this construction is that no zero matrices are introduced between columns of matrix $\mathbf{R}'$.
   As shown in Theorem \ref{SecurThm} below, this construction allows the master server to recover all the desired submatrices $\mathbf{C}_{i,j}$ for $i=1,\ldots,t$ and $j=1,\ldots,d$ from the middle samples of the convolutions $\mathbf{c}_{i,j}=\mathbf{a}_i*\mathbf{b}_j$ (see Fig.~ \ref{figconv} for an illustration).

%
%
 
\begin{Theorem}\label{SecurThm}
For a given security level $P_C<P$, the proposed SGPD code achieves the recovery threshold $P_R$
\begin{equation}\label{RT_1}
{\small
\begin{cases}
tsd+s-1,&\text{ if } P_C=0,\\
t^*s(d+1)+s\Delta_{P_C}-1,& \text{ if } P_C\geq 1 \text{ and }\Delta_{P_C}=\frac{P_C}{s},\\
t^*s(d+1)-s\Delta_{P_C}+2P_C-1,&  \text{ if } P_C\geq 1 \text{ and } \Delta_{P_C}>\frac{P_C}{s},
\end{cases}
}
\end{equation}
and the communication load \eqref{eq_CL},
where  $t^*=t+\Delta_{P_C}$ and $d^*=d+\Delta_{P_C}$ for any integer values $t,s$, and $d$ such that $s< t$, $m=ts$, and $n=sd$.
\end{Theorem}
 \begin{proof}
 The proof is given in Appendix \ref{app}.
 \end{proof}
\begin{Remark}
When $P_C\geq 1$ a direct application of the GPD construction in Fig.~\ref{figGPD} would yield the larger recovery threshold
\begin{equation}
P_R = 
\begin{cases}
t^*sd^*+s-1,& \text{ if }\Delta_{P_C}=\frac{P_C}{s},\\
dst^*+s-1-2(s\Delta_{P_C}-P_C),&  \text{ if } \Delta_{P_C}>\frac{P_C}{s}.
\end{cases}
\end{equation}
\end{Remark}


\subsection{Secure Generalized PolyDot Code: The $s\geq t$ Case}
\begin{figure*}[t!]
	\begin{center}
		\includegraphics[width=9.8cm]{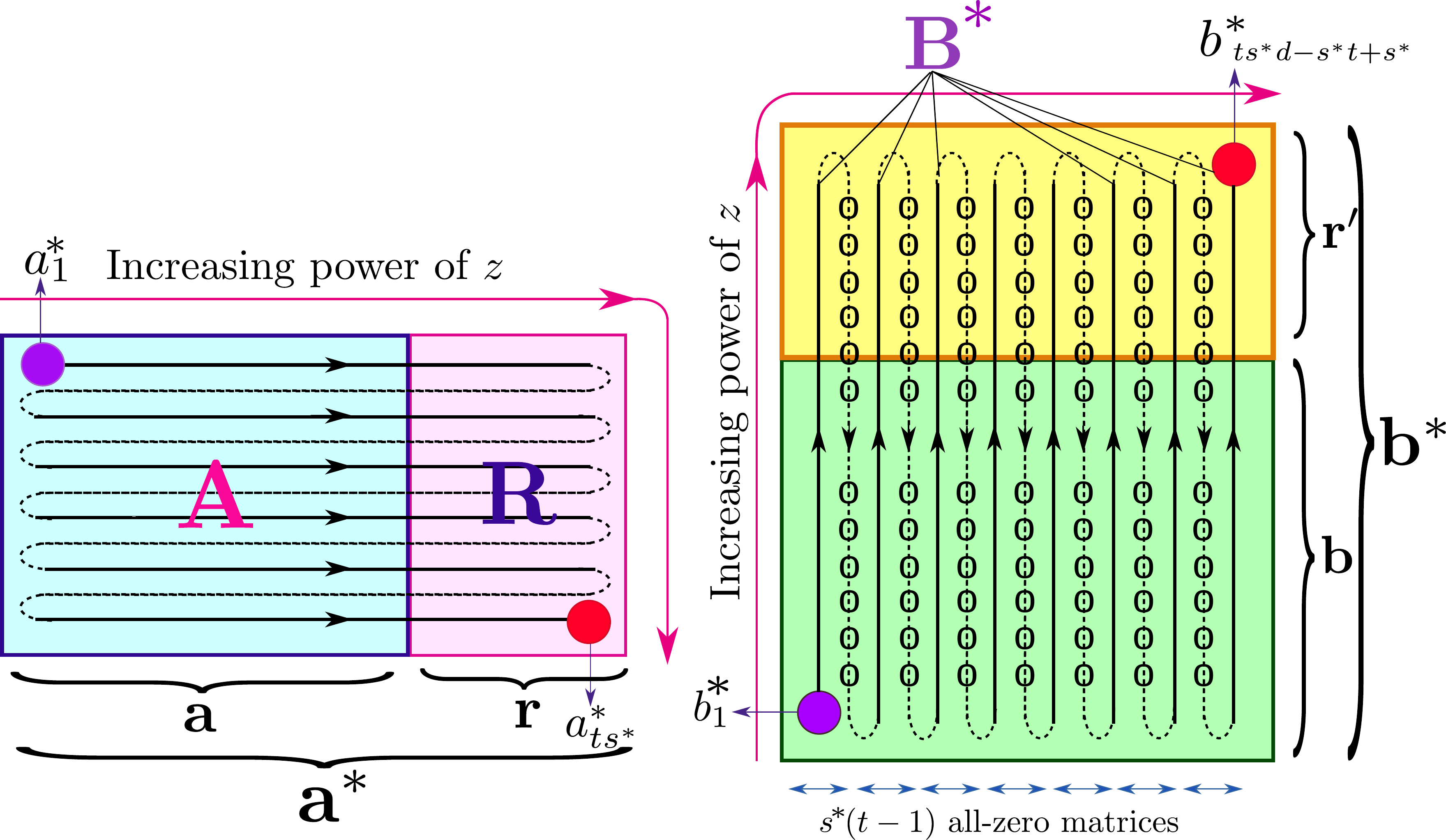}~\caption{\footnotesize
                  Construction of the time block sequences $\mathbf{a}^*$
                  and $\mathbf{b}^*$ in \eqref{eqTimeSeqAR2} and
                  \eqref{eqTimeSeqBB} used to define the SGPD code for  the case $s\geq t$. The solid line and the zero dashed line in $\mathbf{b}^*$ indicate  columns of $\mathbf{B}^*$ and all-zero block sequences, respectively.   }~\label{SecGPD2}
	\end{center}
	\vspace{-5ex}
\end{figure*}
As illustrated in Fig.~\ref{SecGPD2}, when $s\geq t$, we instead augment input matrices $\mathbf{A}$ and $\mathbf{B}$ by adding 
\begin{equation}\label{delta2}
\Delta'_{P_C}\overset{\Delta}{=} \left\lceil \frac{P_C}{\min{\{ t, d \}}} \right\rceil.
\end{equation}
column and row blocks to matrices $\mathbf{A}$ and $\mathbf{B}$. This can be seen to yield a smaller recovery threshold. Accordingly, the $t\times s^*$ augmented block matrix $\mathbf{A}^*=[\mathbf{A}~\mathbf{R}]$ with $s^*=s+\Delta'_{P_C}$, is obtained as 
\begin{equation}\label{eq_matrixA1}
{ 
	\mathbf{A}^*=
	\left[ {\begin{array}{cccccc}
		\mathbf{A}_{1,1}&\ldots&\mathbf{A}_{1,s}&\mathbf{R}_{1,1}&\ldots&\mathbf{R}_{1,\Delta'_{P_C}}\\
		\vdots&\ddots&\vdots&\vdots&\ddots&\vdots\\
		\mathbf{A}_{t,1}&\ldots&\mathbf{A}_{t,s}&\mathbf{R}_{t,1}&\ldots&\mathbf{R}_{t,\Delta'_{P_C}}
		\end{array} } \right] ,}
\end{equation}
while the $s^*\times d$ augmented block matrix $\mathbf{B}^*$ is defined as
\begin{equation}\label{eq_matrixB1}
\mathbf{B}^*=\left[ {\begin{array}{c}
	\mathbf{R}' \\\mathbf{B}
	\end{array} } \right]=
\left[ {\begin{array}{ccc}
	\mathbf{R}'_{\Delta'_{{P_C},1}}&\ldots&\mathbf{R}'_{\Delta'_{{P_C},d}}\\
	\vdots&\ddots&\vdots\\
	\mathbf{R}'_{1,1}&\ldots&\mathbf{R}'_{1,d}\\
	\mathbf{B}_{1,1}&\ldots&\mathbf{B}_{1,d}\\
	\vdots&\ddots&\vdots\\
	\mathbf{B}_{s,1}&\ldots&\mathbf{B}_{s,d}
	\end{array} } \right].
\end{equation}
As for \eqref{eq_matrixA1} and \eqref{eq_matrixB1}, if
$\Delta'_{P_C}-P_C/\min\{t,d\}>0$, the last $s\Delta'_{P_C}-P_C$ block
matrices in \eqref{eq_matrixA1}, with  bottom-to-top right-to-left ordering
in $\mathbf{R}$, and in \eqref{eq_matrixB1} with right-to-left top-to-bottom
ordering in $\mathbf{R}'$, are all-zero block matrices. The construction of
sequences $\mathbf{a}^*$ and $\mathbf{b}^*$ is analogous to the GPD in the non-secure case. In particular, as seen in Fig.~\ref{SecGPD2}, the time block sequence $\mathbf{a}^*$ is  
%
%
 \begin{equation}\label{eqTimeSeqAR2}
 \mathbf{a}^*=\{\mathbf{a}_1,\mathbf{r}_1,\mathbf{a}_2,\mathbf{r}_2,\ldots,\mathbf{a}_d,\mathbf{r}_d \},
 \end{equation}
whereas the block sequence $\mathbf{b}^*$ is defined as 
 \begin{equation}\label{eqTimeSeqBB}
  \mathbf{b}^*=\{\mathbf{b}_1,\mathbf{r}'_1,\mathbf{0},\mathbf{b}_2,\mathbf{r}'_2,\mathbf{0},\ldots, \mathbf{b}_d,\mathbf{r}'_d \}.
  \end{equation}
  Here,  $\mathbf{0}$ is a block sequence of $(t-1)s^*$ all-zero block matrices with dimensions $S/s\times D/d$. 
  

\begin{Theorem}\label{SecurThm2}
For a given security level $P_C<P$, the proposed SGPD code achieves the recovery threshold
	\begin{equation}\label{eq_RT2}
	P_R = 
	\begin{cases}
	s^*(t^2+1)-3 ,& \text{ if }\Delta'_{P_C}>\frac{P_C}{s} \text{ and } t=d  \\
	tds^*+s^*-1,&  \text{ otherwise},
	\end{cases}
	\end{equation}
and the communication load \eqref{eq_CL},
where $s^*=s+\Delta'_{P_C}$ for any integer values $t,s,$ and $d$ such that $s\geq t$, $m=ts$, and $n=sd$.
\end{Theorem}
 \begin{proof}
 The proof is presented in Appendix \ref{app1}.
 \end{proof}
%

\begin{Example}
We now provide some numerical results of the proposed SGPD. We set $P=3000$ and $m=n=36$. The trade-off between communication load $C_L$ and recovery threshold $P_R$ for both non-secure conventional GPD codes $(P_C=0)$ and proposed SGPD code with $P_C=11$ and $P_C=29$ is illustrated in Fig.~\ref{figSecRT_CO}. 
The figure quantifies the loss in terms of achievable pairs $(P_R,C_L)$ that is caused by the security constraint. We also show the performance of the recently proposed $\text{GASP}_{\text{small}}$ scheme in \cite{rafael2018codes}, for $P_C=0$, $P_C=11$, and $P_C=29$.
\end{Example}

\section{Concluding Remarks}
To this best of the authors' knowledge, this work presents the best currently known trade-off curve between communication load and recovery threshold as a function of the desired input privacy level for the problem of distributed matrix multiplication. The main result recovers prior art, including \cite{fahim2017optimal,dutta2018optimal,dutta2018unified}. Among topics for future work, we mention here the establishment of matching converse bounds \cite{yang2019secure} and the consideration of impairments in the communication channel between workers \cite{ha2018wireless}.
\begin{figure}[t!]
	\begin{center}
		\includegraphics[scale=0.75]{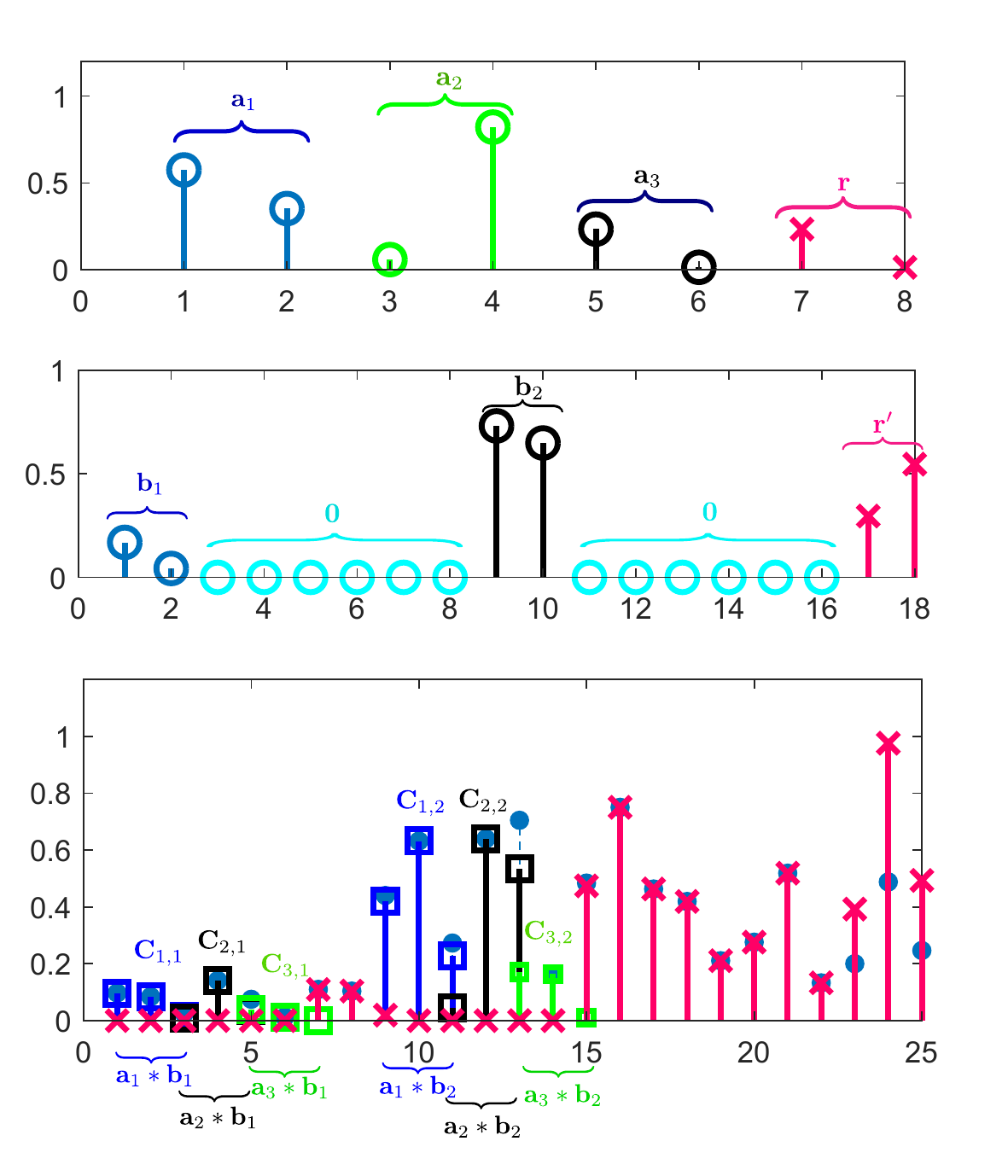}\vspace{-.1cm}~\caption{\footnotesize{Illustration
                    of correct recovery threshold for $t=3,s=2,d=2$, and
                    $P_C=2$. Dashed blue stems with filled markers represent
                    the convolution $\mathbf{c}^*$. Individual convolutions
                    $\mathbf{c}_{i,j}$ are shown in different colors with
                    square markers. Contributions from one or both random
                    matrices are show as red crosses. The desired
                    submatrices $\mathbf{C}_{i,j}$ are seen to equal the
                    corresponding samples from the sequence $\mathbf{c}^*$,
                    associated with the center points of the individual convolutions.}}~\label{figconv}
	\end{center}
	\vspace{-5ex}
\end{figure}

\appendices

\section{Proof of Theorem \ref{SecurThm}}\label{app} 
 
 The $z$-transform of sequences $\mathbf{a}^*$ and $\mathbf{b}^*$ are given respectively as
\begin{align}
   \mathbf{F}_{\mathbf{a}^*}(z)&
   = \underbrace{\sum_{i=1}^{t}\sum_{j=1}^{s} \mathbf{A}^*_{i,j}z^{s(i-1)+(j-1)}}_{\overset{\Delta}{=}~\mathbf F_1(z)}\nonumber\\
   &+\underbrace{\sum_{i=t+1}^{t^*}\sum_{j=1}^{s} \mathbf{A}^*_{i,j} z^{s(i-1)+j-1}}_{\overset{\Delta}{=}~\mathbf F_2(z)},\label{eqsec1_pAM1li}\\
 	\mathbf{F}_{\mathbf{b}^*}(z)&= 	 	 \underbrace{\sum_{k=1}^{s}\sum_{l=1}^{d}\mathbf{B}^*_{k,l}z^{s-k+t^*s(l-1)}}_{\overset{\Delta}{=} ~\mathbf F_3(z)}\nonumber\\
 	&+\underbrace{\sum_{k=1}^{s}\sum_{l=d+1}^{d^*}\mathbf{B}^*_{k,l} z^{t^*sd+s(l-d)-k }}_{\overset{\Delta}{=}~\mathbf F_4(z)}.\label{eqsec1_pBM1li}
 \end{align}
  \begin{figure}[t!]
  	\begin{center}
  		\includegraphics[scale=0.63]{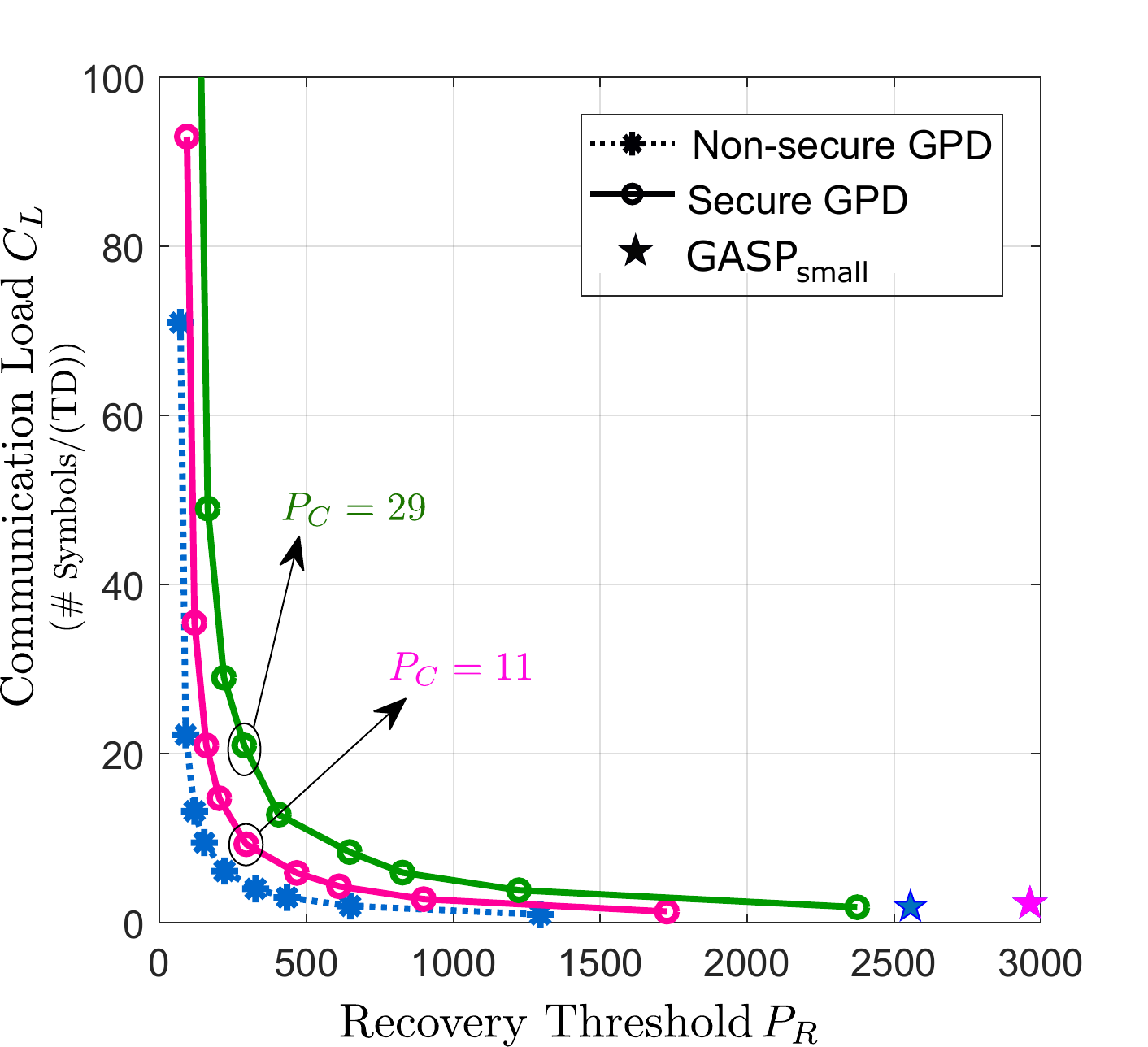}\vspace{-.1cm}~\caption{\footnotesize{Communication
                    load $C_L$ versus recovery threshold $P_R$ for both
                    non-secure generalized PolyDot (GPD) and secure
                    generalized PolyDot (SGPD) codes when $P=3000$ and
                    $m=n=36$. The performance for $\text{GASP}_{\text{small}}$ codes \cite{rafael2018codes} is
                    shown for $P_C=0$ and $P_C=11$. For $P_C=29$ they achieve 
  		$C_L=2.8 , P_R=36291$ (not shown). }}~\label{figSecRT_CO}
  	\end{center}
  	\vspace{-5ex}
  \end{figure}
 The master server evaluates $\mathbf{F}_{\mathbf{a}^*}(z)$ and $\mathbf{F}_{\mathbf{b}^*}(z)$ at $P$ non-zero distinct points $z_1,\ldots,z_P\in \mathbb{F}$, which define the encoding functions, and sends both matrices $\mathbf{A}_p=\mathbf{F}_{\mathbf{a}^*}(z_p)$ and $\mathbf{B}_p=\mathbf{F}_{\mathbf{b}^*}(z_p)$ to worker $p$. Worker $p$ performs the multiplication $\mathbf{F}_{\mathbf{a}^*}(z_p)\mathbf{F}_{\mathbf{b}^*}(z_p)$, and sends the results back to the master server. 
 To reconstruct all blocks $\mathbf{C}_{i,j}$ of matrix $\mathbf{C}=\mathbf{AB}$, the master server carries out polynomial interpolation, or equivalently it computes the inverse $z$-transform, upon  receiving a number of multiplication results equal to at least the length of the sequence $\mathbf{c}^*=\mathbf{a}^**\mathbf{b}^*$. 
As we detail next, the $(i,l)$ block $\mathbf{C}_{i,l}=\sum_{r=1}^{s}\mathbf{A}_{i,r}\mathbf{B}_{r,l}$, for all $i=1,\ldots,t$ and $l=1,\ldots,d$, of matrix $\mathbf{C}=\mathbf{AB}$ can be seen equal to the $(si-1+(l-1)t^*s)$th sample of the convolution $\mathbf{c}^*=\mathbf{a}^**\mathbf{b}^*$. An illustration can be found in Fig.~\ref{figconv}.

To see this, we first note that, by the properties of GPD codes, matrix $\mathbf{C}_{i,l}$ is the coefficient of the monomial $z^{si-1+(l-1)t^*s}$ in $\mathbf{F}_{1}(z)\mathbf{F}_{3}(z)$. Note that this holds since the polynomial $\mathbf{F}_1(z)$ and $\mathbf{F}_3(z)$ are defined as for GPD codes. We now need to show that no other contribution to this term arises from the products $\mathbf{F}_1(z)\mathbf{F}_4(z)$, $\mathbf{F}_2(z)\mathbf{F}_3(z)$, and $\mathbf{F}_2(z)\mathbf{F}_4(z)$. 
The terms in the product $\mathbf{F}_1(z)\mathbf{F}_4(z)$ have exponents $(t^*sd+s(i-1)+s(l-d)-1)$, for $i=1,\ldots,t$ and $l=d+1,\ldots,d^*$, which do not include the desired values $(si-1+(l-1)t^*s)$ for $i=1,\ldots,t$ and $l=1,\ldots,d$.
A similar discussion applies to the product $\mathbf{F}_2(z)\mathbf{F}_3(z)$, whose exponents are $(s(i+t^*l-t^*)-1)$, for $i=t+1,\ldots,t^*$ and $l=1,\ldots,d$, and $\mathbf{F}_2(z)\mathbf{F}_4(z)$, whose exponents are $(t^*sd+s(i-1)+s(l-d)-1)$, for $i=t+1,\ldots,t^*$ and $l=d+1,\ldots,d^*$.

 In order to recover the convolution $\mathbf{c}^*$, the master server needs to collect a number of values of the product $\mathbf{F}_\mathbf{a}(z)\mathbf{F}_{\mathbf{b}}(z)$ equal to the length of the sequence $\mathbf{c}^*$, which can be computed as the degree $\deg \left(\mathbf{F}_\mathbf{a}(z)\mathbf{F}_{\mathbf{b}}(z)\right)+1$, where $\deg (\mathbf{F}_\mathbf{a}(z)\mathbf{F}_{\mathbf{b}}(z))$ is 
\begin{equation}\label{eq_deg}
 \begin{cases}
 t^*s(d+1)+s\Delta_{P_C}-1,& \text{ if }\Delta_{P_C}=\frac{P_C}{s},\\
 dst^*-s\Delta_{P_C}+2P_C+t-2,& \text{ if }\Delta_{P_C}>\frac{P_C}{s},
 \end{cases}
\end{equation}
	  which for $P_C\geq 1$ implies the recovery threshold $P_R$ in \eqref{RT_1}. 
	  The communication load $C_L$ in \eqref{eq_CL} follows from the fact that there are $TD/(td)$ entries in $\mathbf{F}_{\mathbf{a}^*}(z_p)\mathbf{F}_{\mathbf{b}^*}(z_p)$, for all $p\in [1,P_R]$.  
	  
The security constraint \eqref{eq_security_constraint} can be proved in a manner similar to \cite{chang2018capacity} by the following steps:
\begin{align}\label{eq_sec}
&I(\mathbf{A},\mathbf{B};\mathbf A_{\mathcal{P}},\mathbf B_{\mathcal{P}})\nonumber\\
=&H(\mathbf A_{\mathcal{P}},\mathbf B_{\mathcal{P}})-H(\mathbf A_{\mathcal{P}},\mathbf B_{\mathcal{P}}|\mathbf{A},\mathbf{B})\nonumber\\
\overset{(a)}{=}&H(\mathbf A_{\mathcal{P}},\mathbf B_{\mathcal{P}})-H(\mathbf A_{\mathcal{P}},\mathbf B_{\mathcal{P}}|\mathbf{A},\mathbf{B})\nonumber\\
&+H(\mathbf A_{\mathcal{P}},\mathbf B_{\mathcal{P}}|\mathbf{A},\mathbf{B},\mathbf{R}_1,\ldots,\mathbf{R}_{P_C},\mathbf{R}'_1,\ldots,\mathbf{R}'_{P_C})\nonumber\\
\overset{}{=}&H(\!\mathbf A_{\mathcal{P}},\mathbf B_{\mathcal{P}}\!)-I(\!\mathbf A_{\mathcal{P}},\mathbf B_{\mathcal{P}};\mathbf{R}_1,\ldots,\mathbf{R}_{P_C},\mathbf{R}'_1,\ldots,\mathbf{R}'_{P_C}|\mathbf{A},\mathbf{B})\nonumber\\
\overset{}{=}&H(\mathbf A_{\mathcal{P}},\mathbf B_{\mathcal{P}})-H(\mathbf{R}_1,\ldots,\mathbf{R}_{P_C},\mathbf{R}'_1,\ldots,\mathbf{R}'_{P_C}|\mathbf{A},\mathbf{B})\nonumber\\&+H(\mathbf{R}_1,\ldots,\mathbf{R}_{P_C},\mathbf{R}'_1,\ldots,\mathbf{R}'_{P_C}|\mathbf{A},\mathbf{B},\mathbf A_{\mathcal{P}},\mathbf B_{\mathcal{P}})\nonumber\\
\overset{(b)}{=}& 
H(\mathbf A_{\mathcal{P}},\mathbf B_{p})-H(\mathbf R_{1},\ldots,\mathbf{R}_{P_C},\mathbf R'_{1},\ldots,\mathbf{R}'_{P_C})\nonumber\\
\overset{(c)}{\leq}&H(\mathbf A_{\mathcal{P}})+H(\mathbf B_{\mathcal{P}})-\sum_{p=1}^{P_C}H(\mathbf{R}_p)-\sum_{p=1}^{P_C}H(\mathbf{R}'_p)\nonumber\\
\overset{(d)}{=}&H(\mathbf A_{\mathcal{P}})+H(\mathbf B_{\mathcal{P}})-P_C\frac{TS}{m}\log|\mathbb{F}|-P_C\frac{SD}{n}\log|\mathbb{F}|\nonumber\\
\overset{(e)}{\leq}&\sum_{p=1}^{P_C}H(\mathbf{A}_p)+\sum_{p=1}^{P_C}H(\mathbf{B}_p)-P_C\frac{TS}{m}\log|\mathbb{F}|-P_C\frac{SD}{n}\log|\mathbb{F}|\nonumber\\
\overset{(f)}{=}&P_C\frac{TS}{m}\log|\mathbb{F}|+P_C\frac{SD}{n}\log|\mathbb{F}|-P_C\frac{TS}{m}\log|\mathbb{F}|\nonumber\\
&-P_C\frac{SD}{n}\log|\mathbb{F}|\nonumber\\
=&0,
\end{align}
where $(a)$ follows from the definition of encoding functions, since $\mathbf{A}_{\mathcal{P}}$ is a deterministic function of $\mathbf{A}$ and $\mathbf{R}_p$ and $\mathbf{B}_{\mathcal{P}}$ is a deterministic function of $\mathbf{B}$ and $\mathbf{R}'_p$, respectively, for all $p=1,\ldots,P_C$; 
$(b)$ follows from \eqref{eqsec1_pAM1li} and \eqref{eqsec1_pBM1li}, since from $P_R$ polynomial evaluations $\mathbf{A}_{\mathcal{P}}$ and $\mathbf{B}_{\mathcal{P}}$ in \eqref{eqsec1_pAM1li} and \eqref{eqsec1_pBM1li} we can recover $2P_C$ unknowns when the coefficients $\mathbf{A}_{i,j}$ and $\mathbf{B}_{k,l}$ are known, given that we have $P_R\geq 2P_C$; $(c)$ and $(d)$ follow since $\mathbf{R}_p$ and $\mathbf{R}'_p$ are independent uniformly distributed entries; 
$(e)$ follows by upper bounding the joint entropy using the sum of individual entropies; and $(f)$ follows from an argument similar to $(d)$. Hence, the proposed scheme is information-theoretically secure.

\section{Proof of Theorem \ref{SecurThm2}}\label{app1}
 
We define the $z$-transform of sequences $\mathbf{a}^*$ and $\mathbf{b}^*$ respectively as
	\begin{align}
	\mathbf{F}_{\mathbf{a}^*}(z)
	&=\sum_{i=1}^{t}\sum_{j=1}^{s} \mathbf{A}^*_{i,j} z^{s^*(i-1)+j-1}\nonumber\\
	&+\sum_{i=1}^{t}\sum_{j=s+1}^{s^*} \mathbf{A}^*_{i,j} z^{s^*(i-1)+j-1},\label{eqsec2_pAM1li}\\
		\mathbf{F}_{\mathbf{b}^*}(z)
		&=\sum_{k=1+\Delta_{P_C}'}^{s^*}\sum_{l=1}^{d}\mathbf{B}^*_{k,l}z^{s^*-k+ts^*(l-1)}\nonumber\\
		&+ \sum_{k=1}^{\Delta_{P_C}'}\sum_{l=1}^{d}\mathbf{B}^*_{k,l} z^{s^*-k+ts^*(l-1)} \label{eqsec2_pBM1li}.
 \end{align}
 The $(i,l)$ block $\mathbf{C}_{i,l}=\sum_{r=1}^{s}\mathbf{A}_{i,r}\mathbf{B}_{r,l}$, for all $i=1,\ldots,t$ and $l=1,\ldots,d$, of matrix $\mathbf{C}=\mathbf{AB}$ can be seen equal to the $(s^*i-1+(l-1)ts^*)$th sample of the convolution $\mathbf{c}^*=\mathbf{a}^**\mathbf{b}^*$.
 The rest of the proof follows in a manner akin to Theorem \ref{SecurThm}.

  \ifCLASSOPTIONcaptionsoff
 	\newpage
 	\fi


 	
 	%
 	
 	
 	 \balance
 	\bibliographystyle{IEEEtran} 
 	\bibliography{IEEEabrv,references} 

 \end{document}